\documentclass[10pt]{article}
\usepackage{spconf,amsmath,graphicx,mathtools,flexisym,array,float,multirow,epstopdf}
\usepackage{upgreek}
\newcolumntype{P}[1]{>{\centering\arraybackslash}p{#1}}
\usepackage{amsmath}
\usepackage{amsthm}
\usepackage{dsfont}
\usepackage{thmtools}
\usepackage{amssymb}
\usepackage[dvipsnames]{xcolor}
\usepackage[utf8]{inputenc}
\usepackage{comment}
\usepackage[english]{babel}
\usepackage{color}
\usepackage{framed}
\usepackage{cite}
\usepackage{amsmath,amssymb,amsfonts}
\usepackage{IEEEtrantools}
\usepackage{graphicx}
\usepackage{upgreek}
\usepackage{gensymb}
\usepackage{textcomp}
\usepackage{cite}
\usepackage{bbold}
\usepackage{ifpdf}
\usepackage{times}
\usepackage{layout}
\usepackage{float}
\usepackage{afterpage}
\usepackage{amsmath}
\usepackage{amstext}
\usepackage{amssymb,bm}
\usepackage{latexsym}
\usepackage{color}
\usepackage{kantlipsum}    
\usepackage{graphicx}
\usepackage{amsmath}
\usepackage{amsthm}
\usepackage{graphicx}
\usepackage{url}
\usepackage[font=scriptsize,caption=false,labelsep=space]{subfig}
\usepackage{booktabs}
\usepackage{multicol}
\usepackage{lipsum}% dummy text

\usepackage{enumitem}
\usepackage[switch]{lineno}
\usepackage[named]{algo}
\usepackage[noend]{algpseudocode}
\usepackage{algorithm}
\usepackage{seqsplit}
\usepackage{arydshln}
\usepackage{mathtools}

\def\bx{\boldsymbol{x}}

\def\bSigma{\boldsymbol{\Sigma}}

\def\bPhi{\boldsymbol{\Phi}}

\def\mba{\mathbf{a}}
\def\mbb{\mathbf{b}}

\def\mbw{\mathbf{w}}
\def\mbx{\mathbf{x}}
\def\mby{\mathbf{y}}

\def\mbA{\mathbf{A}}

\def\mbH{\mathbf{H}}
\def\mbI{\mathbf{I}}

\def\mbW{\mathbf{W}}
\def\mbX{\mathbf{X}}
\def\mbY{\mathbf{Y}}

\def\calA{\mathcal{A}}

\def\calH{\mathcal{H}}

\def\calM{\mathcal{M}}

\def\calR{\mathcal{R}}
\def\calS{\mathcal{S}}
\def\calT{\mathcal{T}}
\def\calU{\mathcal{U}}

\def\bzero{\boldsymbol{0}}

\newcommand{\complexC}[1]{\mathds{C}^{#1}}
\newcommand{\expecE}[1]{\mathds{E}\left\{{#1}\right\}}

\newcommand{\realR}[1]{\mathds{R}^{#1}}

\def\Diag#1{\mathrm{Diag}\left(#1\right)}

\newtheorem{theorem}{Theorem}

\def\T{\top}
\def\H{\mathrm{H}}

\let\emptyset\varnothing

\newcommand*{\rom}[1]{\expandafter\@slowromancap\romannumeral #1@}
\usepackage{tikz}
\title{Submodular Optimization for Placement of Intelligent Reflecting Surfaces in  sensing systems}
\name{Zahra Esmaeilbeig$^1$, 
 Kumar Vijay Mishra$^{2}$, Arian Eamaz$^1$, and Mojtaba Soltanalian$^1$
\thanks{This work was sponsored in part by the Army Research Office, accomplished under Grant Number W911NF-22-1-0263, and in part by the Department of Navy award N00014-22-1-2666 issued by the  Office of Naval Research. Any opinions, findings, conclusions, or recommendations expressed in this material are those of the authors and should not be interpreted as representing the official policies, either expressed or implied, of the Army Research Office, the Office of Naval Research, or the US Government. The U.S. Government is authorized to reproduce and distribute reprints for
Government purposes notwithstanding any copyright notation herein.}}
\address{$^{1}$ECE Department, University of Illinois Chicago, USA\\
 $^{2}$United States DEVCOM Army Research Laboratory, USA\vspace{-16pt}}
\ninept
\begin{document}
% Reduce spacing above and below equations
\setlength{\abovedisplayskip}{3pt}
\setlength{\belowdisplayskip}{3pt}
\maketitle
\begin{abstract}
Intelligent reflecting surfaces (IRS) are increasingly being investigated for novel sensing applications for improving target estimation and detection. While the optimization of IRS phase-shifts has been studied extensively, the optimal placement of multiple IRS  platforms for sensing applications is relatively unexamined. In this paper, we focus on selecting optimized locations of IRS platforms by considering a mutual information (MI)-based criterion within the radar coverage area. 
We demonstrate the submodularity of the MI criterion and then tackle the design problem by means of a constant-factor approximation algorithm for submodular optimization. Numerical results validate the proposed submodular optimization framework for optimal IRS placement with the worst-case performance bounded to $1-1/e\approx 63 \%$.
 \end{abstract}
\vspace{-6pt}
\vspace{6pt}
\section{Introduction}
The performance of wireless communications and sensing systems is severely
reduced in poor channel conditions and blockage~\cite{elbir2022rise,wei2023ris}. % Chepuri2023}. 
Recently, the deployment of an intelligent reflecting surface (IRS) in the channel has been investigated to enhance the system performance for blocked or high-noise channels. In wireless communications, IRS has shown improvement in channel capacity \cite{an2023stacked} and security~\cite{tyrovolas2022performance}, interference mitigation \cite{wang2023}, and integrated sensing and communications \cite{elbir2022rise,esmaeilbeig2023quantized,wang2023stars,wei2022simultaneous,wei2023multi} and index modulation~\cite{hodge2023}. Similar gains for radar systems have been demonstrated by optimizing IRS  phase-shifts for target estimation~\cite{esmaeilbeig2022irs,esmaeilbeig2022joint,esmaeilbeig2022cramer} and detection~\cite{esmaeilbeig2023moving}.

The optimal placement of IRSs in the coverage area plays a crucial role in the performance of wireless systems. In IRS-aided wireless communications, there have been studies on the optimal placement of IRS for optimizing the throughput \cite{Huangplacement2022} and coverage \cite{zeng2020reconfigurable}. In \cite{stratidakis2022optimal}, the positioning of IRS platforms in mobile user environments is achieved by imposing a desired minimum threshold on the received power at users. In \cite{Ghose2023}, the optimal distance between the IRS and users in a device-to-device communications system was derived analytically. Other placement criteria such as maximization of receive signal-to-noise-ratio (SNR) \cite{ntontin2021optimal} and blind zone improvement \cite{chen2022reconfigurable} have also been considered for optimal IRS placement in millimeter-wave communications scenario.
 
In general, prior works on IRS-aided radars keep the deployment location of IRSs fixed and then focus on optimizing IRS phase-shifts. Recently, a near-optimal placement was investigated in~\cite{tohidi2023near} to maximize the extended line-of-sight (LoS) coverage provided by IRS platforms in an urban infrastructure. However, this study treated the IRS platforms as specular surfaces \cite{ulaby1981microwave}, thereby implying that the optimal IRS placement is determined irrespective of the IRS phase-shifts. 
 
 Contrary to prior works, we address the problem of optimal placement of IRS platforms tuned with predetermined phase-shifts. %We leave the joint IRS placement and  phase-shift design to be explored in a future study. To design the placement, 
 We employ the mutual information (MI) between the channel state information (CSI) modified by IRS phase-shifts and the received signal. Previous works have employed MI for designing radar systems and signals~\cite{naghsh2013majorization,naghsh2014unimodular,radarsignaldesign2022,liu2020co}. %We use MI for IRS placement, leading to a combinatorial optimization problem that is NP-hard in nature.
 The IRS placement is a combinatorial search problem, wherein the best placement for IRS platforms is determined exhaustively from a finite set of possibilities. In the context of sensor selection, combinatorial search has been avoided by relaxing the problem to convex optimization \cite{joshi2009sensor}. Similarly, \cite{Godrich} formulates the selection of the sensors as a \textit{knapsack problem} for a distributed multiple-radar scenario and solves it through greedy search with the Cram\'{e}r-Rao lower bound (CRB) as a performance metric. More recently, deep learning methods~\cite{mishra2023sparse} have been employed. 
 
 A framework for combinatorial optimization with optimality guarantees is, indeed, of broad interest. We propose to employ submodular optimization to obtain optimal placement of IRSs with a proper optimality guarantee. Previously, sensor placement in ~\cite{shulkind2018sensor, tohidi2020submodularity} has shown that submodular optimization provides a greedy algorithm for \textit{constant-factor approximation} of the optimal solution. The approximation factor is dependent on the curvature of the submodular function. The solution obtained at the worst-case execution of the greedy algorithm is only up to a known factor of the optimal solution. In this paper, we prove the monotonicity and submodularity of the MI-based cost function for optimal IRS placement. Then, we develop a constant-factor approximation algorithm to decide the IRS locations in a greedy manner by exploiting the submodularity of the cost function. Our numerical experiments illustrate the superiority of the IRS placements obtained via the greedy algorithm in comparison to a random placement. We also compare our proposed method with the theoretical lower bound on the worst-case performance of the greedy algorithm. % which is computed as a function of the  curvature of the   submodular objective function,  

Throughout this paper, we use bold lowercase and bold uppercase letters for vectors and matrices, respectively. 
$(\cdot)^{\mathrm{H}}$ is the vector/matrix Hermitian transpose; The determinant of a matrix is denoted by $\operatorname{det}(.)$; the function $\textrm{diag}(.)$ returns the diagonal elements of the input matrix, while $\textrm{Diag}(.)$ produces a diagonal matrix with the same diagonal entries as its vector argument.
Given a reference set $\calA$ and a subset
$\calU \subseteq \calA$, the absolute complement  of $\calA$ in relation to $\calU$ is  denoted by $\calU \setminus \calA$. $|\calA|$ is the cardinality of the set $\calA$. The positive semidefiniteness of the matrix $\mbA$ is denoted by $\mbA \succeq 0$. Finally, $\mbA=\left[\begin{array}{c|c|c}
\mbA_{_{1}} &\ldots&\mbA_{_{M}}\end{array}\right]$ is a   matrix  made by   stacking the  blocks  $\mbA_{_{1}}, \ldots, \mbA_{_{M}}$ horizontally.%\vspace{6pt}
\section{Signal Model}%\vspace{6pt}
Consider a multi-IRS-assisted radar system (Fig.1), which comprises
a co-located MIMO radar with $N_t$ transmit and $N_r$ receive antennas, each arranged as uniform linear arrays (ULA) with inter-element spacing $d$. We deploy $M$ IRS platforms indexed as IRS$_{1}$, IRS$_{2}$, $\ldots$, IRS$_{M}$, each equipped with $N_m$ reflecting elements arranged as ULA. Denote the $N_t \times 1$ vector of all  transmit signals as $\bx(t)=[x_{_{1}}(t),\ldots,x_{_{N_t}}(t)]^{\T}$, where the continuous-time signal transmitted from the $n$-th antenna at time instant $t$ is $x_n(t)$. The steering vectors of the radar transmitter, receiver, and the $m$-th IRS are, respectively, 
\begin{align}
\mba_t(\theta)&=[1,e^{\textrm{j}\frac{2\pi}{\lambda}d sin \theta},\ldots,e^{\textrm{j}\frac{2\pi}{\lambda}d(N_t-1) sin \theta}]^{\top}, \\
\mba_r(\theta)&=[1,e^{\textrm{j}\frac{2\pi}{\lambda}d sin \theta},\ldots,e^{\textrm{j}\frac{2\pi}{\lambda}d(N_r-1) sin \theta}]^{\top}, and  \\
\mbb_m(\theta)&= [1,e^{\textrm{j}\frac{2\pi}{\lambda}d_m sin \theta},\ldots,e^{\textrm{j}\frac{2\pi}{\lambda}d_m(N_m-1) sin \theta}]^{\top},
\end{align}
where $\lambda$ is the carrier wavelength; $d$ and $d_m$ are inter-element spacings in the transceiver and IRS arrays, respectively. While $d$ is usually set to be half of the carrier wavelength, $d_m$ is much smaller at the subwavelength level. The IRS platforms are mounted on static platforms, whose optimal locations we intend to determine later.

Denote the angle between the radar-target, radar--IRS$_m$, and  target-IRS$_m$ by $\theta_{tr}$, $\theta_{ri,m}$, and $\theta_{ti,m}$, respectively. For the multi-IRS aided radar the non-line-of-sight (NLoS) channels  associated with IRS$_m$ are  $\mbH_{ri,m}=\mbb_m(\theta_{ri,m}) \mba_{t}^{\top}(\theta_{ri,m})\in \mathbb{C}^{N_m \times N_t}$ for radar-IRS$_{m}$; 
$\mbH_{it,m}=\mbb_{m}^{\top}(\theta_{ti,m}) \in \mathbb{C}^{1 \times N_m}$ for IRS$_{m}$-target; 
$\mbH_{ti,m}=\mbb_m(\theta_{ti,m})\in \mathbb{C}^{N_m\times 1}$ for target-IRS$_{m}$; and
$\mbH_{ir,m}=\mba_r(\theta_{ri,m}) \mbb^{\top}_m(\theta_{ri,m}) \in \mathbb{C}^{N_r\times N_m}$ for IRS$_{m}$-radar paths. Assume the $k$-th passive element of IRS$_{m}$ is tuned at the phase-shift $\phi{_{m,k}}\in[0,2\pi)$. Then, IRS$_{m}$ is characterized by the phase-shift matrix 
\begin{equation}
\bPhi_{m}=\Diag{[e^{\textrm{j}\phi_{_{m,1}}},\ldots,e^{\textrm{j}\phi_{_{m,N_m}}}]}\; \in \complexC{N_m \times N_m}.
\end{equation}

Consider a target located at a distance $d_{tr}$ with respect to (w.r.t.) the radar transceiver and complex reflectivity  $\alpha_{_{{m}}}$, which depends on the target size, target shape, and propagation path loss. The delay of the signal propagated in the radar-IRS$_{m}$-target-IRS$_{m}$-radar path is $\tau_{ritir,m}$. The line-of-sight (LoS) or radar-target-radar path is obstructed (Fig. 1). The transmit signal backscattered by the target and received by all $N_r$ receive antennas is 
\begin{align}
 \mby(t)&= \sum_{m=1}^{M}\alpha_{_{m}}\mbH_{ir,m}\bPhi_m\mbH_{ti,m} \mbH_{it,m}\bPhi_m\mbH_{ri,m}\bx(t-\tau_{ritir,m})\nonumber \\
& \hspace{5cm} +\mbw(t),\;\; \in \complexC{N_r \times 1}.
\end{align}
where $\mbw(t)\sim \mathcal{CN}(\bzero,\sigma^2\mbI_{N_t})$ denotes a stationary additive white Gaussian noise (AWGN). The relative time gaps between any two multipath signals are very small in comparison to the actual roundtrip delays, i.e., $\tau_{ritir,m} \approx \tau_0 =\frac{2d_{tr}}{c}$ for $m \in \{1,\ldots,M\}$ and  $c$ is the speed of light. Denote $$\mbH_m=\mbH_{ir,m}\bPhi_{m}\mbH_{ti,m} \mbH_{it,m}\bPhi_m\mbH_{ri,m}\in\mathbb{C}^{N_r\times N_t}.$$

We collect $N$  samples at the rate $1/T_s$ from $\mby(t)$, corresponding to the range cell of a notional target located at range $d_{tr}$,  The $N_r \times 1$ discrete-time received signal vector $\mby[n]= \mby(nT_s)$ is 
\begin{equation}
\label{eq_1}
\mby[n]=\sum_{m=1}^{M} \alpha_{_{m}} \mbH_{m} \mbx[n]+\mbw[n],\; n=0,\ldots,N-1,
\end{equation}\normalsize
where  
$\mbx[n]=\mbx(\tau_0+nT_s)\in \mathbb{C}^{N_t \times 1}$. The delay $\tau_0$ is aligned on-the-grid so that $n_0=\tau_0/T_s$ is an integer~\cite{mishra2017sub}. Define $\overline{\mbH}=\left[\begin{array}{c|c|c}
\alpha_{_1}\mbH_{_{1}}^{\H} &\ldots&\alpha_{_M}\mbH_{_{M}}^{\H}
\end{array}\right]^{\H} \in \mathbb{C}^{N_t \times MN_r}$ and rewrite the discrete-time $N_r\times 1$ received signal vector as 
\begin{equation}
\mby[n]=\overline{\mbH}\mbx[n] +\mbw[n],\; n=0,\ldots,N-1,
\end{equation}
where $\mbw[n]$ is the circularly symmetric complex white Gaussian noise vector, with a zero mean and variance $\sigma^2 \mbI_{N_r}$. Stacking $N$ discrete-time samples for all $N_r$  receiver antennas, the received signal is the $N_r \times N$ matrix 
 $\mbY=\left[\begin{array}{c|c|c}\mby[0]&\ldots&\mby[N-1]\end{array}\right]$. i.e., 
\begin{equation}\label{eq:model}
\mbY=\overline{\mbH} \mbX+\mbW,
\end{equation}
where $\mbX=\left[\begin{array}{c|c|c} \mbx[0]&\ldots&\mbx[N-1]\end{array}\right]\in \mathbb{C}^{N_t \times N}$ 
 and $\mbW=\left[\begin{array}{c|c|c}\mbw[0]&\ldots &\mbw[N-1]\end{array}\right]\in \mathbb{C}^{N_r \times N}$.
 
 In general, the placement of the IRS affects the SNR in the received signal. In the next section, we investigate the optimal placement of the IRS platforms based on a mutual information criterion. %  we consider for the signal model~\eqref{eq:model}.
%----------------------------------
\begin{figure}[t]\label{fig_1}
\centering
\input{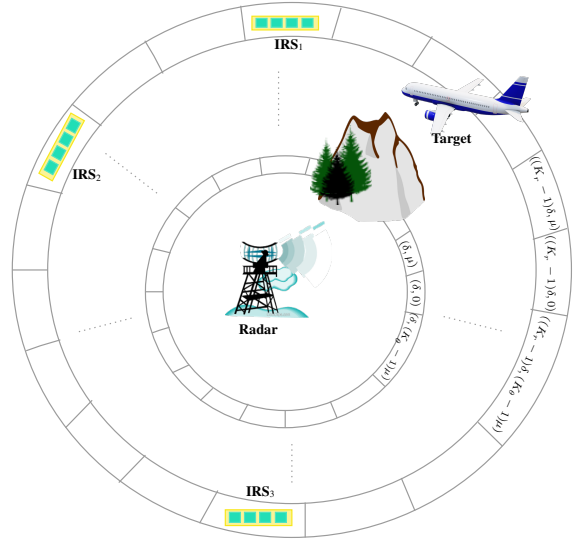}
\caption{Simplified illustration of IRS-aided radar system and range-azimuth bins $(r,\theta)\in \calA$. When the LoS is blocked, the target is sensed by the radar through the NLoS paths via the IRSs.} 
\end{figure}
%----------------------------------

\section{IRS Placement}
%\textcolor{red}{give a brief intro about why placement is needed. what are other ways in the literature for sensor placement? Then, say, why do you use MI}
Our criterion for optimizing the placement of IRS platforms is the MI
\begin{align}
\mathrm{I}\left(\mbY;\overline{\mbH}\bigg| \mbX\right)&=  \calH\left(\mbY \bigg| \mbX\right)-\calH\left(\mbY \bigg| \overline{\mbH}, \mbX\right)\nonumber\\
&=\calH\left(\mbY \bigg| \mbX\right)-\calH\left(\mbW\right)
\end{align}
where $\calH(\cdot)$ is the Shannon entropy.  Discretize the range-azimuth plane of the radar coverage area, within which the IRSs are installed, with $K$ ranges with a spacing of $\delta$ denoted by the set $\calR = \{\delta,2\delta,\ldots,(K_{_{r}}-1)\delta\}$ and with $K$ azimuth angles with a spacing of $\mu$  denoted by the set $\mathcal{M} = \{\mu,2\mu,\ldots,(K_{_{\theta}}-1)\mu\}$; see Fig.1. Define the range-azimuth set $\calA$ as 
\begin{align}
\calA=\calR  \times \mathcal{M} = \left \{(r,\theta)\mid r \in \calR, \theta \in \mathcal{M} \right \}.
\end{align}

Denote the channel matrix contributed by the set of IRS platforms implemented in two-dimensional (2-D) locations in the set  $\calS$ by  $\overline{\mbH}_{_{\calS}}$. Then, we choose $\mathcal{S}$ by solving the optimization problem %\textcolor{red}{too direct. no discussion of what is the design problem, what is the goal, etc. First, explain the goal and then say it can be achieved by solving the following optimization problem because .... Also, no explanation of any property of MI and why it would help.}
\begin{align}
\label{eq:opt1}
\mathcal{P}_{1}:~\underset{\calS \subseteq \calA}{\textrm{maximize}} &\quad \mathrm{I}\left(\mbY;\overline{\mbH}_{_{\calS}}\bigg| \mbX\right) \nonumber \\
\text{subject to} &\quad |\calS| \leq  M.
\end{align}
Note that monotone submodular maximization is nontrivial only
when constrained. Our goal is to deploy at most $M$ IRS platforms in the radar aperture  $\calA$. Theorem~\ref{theorem1} below demonstrates the submodularity and monotonicity of the objective function in~\eqref{eq:opt1}, which we then employ for a greedy algorithm with constant-factor approximation error to tackle the problem.
\begin{theorem} The conditional entropy of  $\mbY$ given $\mbX$ is 
\label{theorem1}
\begin{align}
\calH\left(\mbY \bigg| \mbX\right)
&=N\frac{\mathrm{ln}\left(\mathrm{det}(\overline{\mbH} \bSigma_{_{\mbX}}\overline{\mbH}^{\H}+\sigma^2 \mbI)\right)}{2}\nonumber\\
&+\frac{NN_r}{2}\left(1+\mathrm{ln}(2\pi)\right),  
\end{align}
where  $\bSigma_{_{\mbX}}=\expecE{\mbX \mbX^{H}}$ is the total transmit power.
\end{theorem}
\begin{IEEEproof} Considering  that $\mby[n]$ is a  multivariate complex Gaussian distribution,  the  conditional entropy is 
\begin{align}
\calH\left(\mbY \bigg| \mbX\right)&=\calH\left(\mby[0],\ldots, \mby[N-1] \bigg| \mbx[0],\ldots,\mbx[N-1]\right)\nonumber\\
&= \sum_{n=0}^{N-1} \calH\left (\mby[n] \bigg| \mbx[n]\right )=N\frac{\mathrm{ln}\left(\mathrm{det}(\overline{\mbH} \bSigma_{_{\mbX}}\overline{\mbH}^{\H}+\sigma^2 \mbI)\right)}{2} \nonumber \\&+\frac{NN_r}{2}\left(1+\mathrm{ln}(2\pi)\right),
\end{align}
where the second equality follows from the chain rule for entropy and independence of samples and the last equality is based on the entropy of a multivariate Gaussian distribution \cite{cover1999elements}. 
\end{IEEEproof}

 For notational simplicity, denote $\bSigma_x=P_{_{T}} \mbI $ and $\mbH_{_{\calS}}=\frac{\sqrt{P_{_T}}}{\sigma} \overline{\mbH}_{_{\calS}}$, and substituting the result of  Theorem 1 in $\mathcal{P}_{1}$, we  reach at the equivalent problem
\begin{align} \label{eq:opt3}
\mathcal{P}_{2}:~\underset{\calS \subseteq \calA}{\textrm{maximize}} &\quad \mathrm{ln}\left(\mathrm{det}(\mbH_{_{\calS}} \mbH_{_{\calS}}^{\H}+ \mbI)\right) \nonumber \\
\text{subject to} &\quad |\calS| \leq  M.
\end{align}
We aim to tackle the problem $\mathcal{P}_{2}$ through submodular optimization. To this end, we show in Theorem~\ref{theorem2} below that the objective function in $\mathcal{P}_{2}$ is monotonic and submodular.
\begin{theorem}
\label{theorem2}
  The set function  $f(\calS):2^{\calS} \rightarrow \realR{}$, where $f(\calS)=\mathrm{ln}\left(\mathrm{det}(\mbH_{_{\calS}} \mbH_{_{\calS}}^{\H}+ \mbI)\right)$ is a monotonically increasing and submodular function of the set  $\calS$. 
\end{theorem}
\begin{proof}
For  $\calS \subseteq \calT \subseteq \calA $, a  monotone  set function  should satisfy 
\begin{equation}
 f(\calS \cup \calT) \geq f(\calS). 
\end{equation}
The function $f(\calS)$ is submodular if it satisfies  the  diminishing returns property~\cite{tohidi2020submodularity}, i.e., $\forall \, u \in \calA \setminus \calT$,
\begin{equation}
f(\calS \cup \{u\}) - f(\calS) \geq f(\calT \cup \{u\}) - f(\calT).
\end{equation}
In the  context of our  problem, we have 
\begin{align}\label{monotone}
 f(\calS \cup \calT)-f(\calS)&= \mathrm{ln}\left(\frac{\mathrm{det}(\mbH_{_{\calS \cup\calT}} \mbH_{_{\calS \cup\calT}}^{\H}+ \mbI)}{\mathrm{det}(\mbH_{_{\calS}} \mbH_{_{\calS}}^{\H}+ \mbI)}\right)\nonumber\\
&= \mathrm{ln}\left(\frac{\mathrm{det}(\mbH_{_{\calS \cup\calT}}^{\H} \mbH_{_{\calS \cup\calT}}+ \mbI)}{\mathrm{det}(\mbH_{_{\calS}}^{\H} \mbH_{_{\calS}}+ \mbI)}\right)\nonumber\\
&=\mathrm{ln}\left(\frac{\mathrm{det}(\mbH_{_{\calS }}^{\H} \mbH_{_{\calS }}+\mbH_{_{\calT}}^{\H} \mbH_{_{\calT}}+ \mbI)}{\mathrm{det}(\mbH_{_{\calS}}^{\H} \mbH_{_{\calS}}+ \mbI)}\right) \geq 0,
\end{align}
where the second equality follows from, without loss of  generality, $\mbH_{_{\calS \cup \calT}}=\left[\begin{array}{c|c}
\mbH_{_{\calS }}^{\H} &  \mbH_{_{\calT}}^{\H}
\end{array}\right]^{\H}$ leading to $\mbH_{_{\calS \cup\calT}}^{\H} \mbH_{_{\calS \cup\calT}}=\mbH_{_{\calS }}^{\H} \mbH_{_{\calS }}+\mbH_{_{\calT}}^{\H} \mbH_{_{\calT}}$. Additionally, positive semidefiniteness of $\mbH_{_{\calS }}^{\H} \mbH_{_{\calS }}$ i.e. $\mbH_{_{\calS }}^{\H} \mbH_{_{\calS }} \succeq \bzero $ leads to $\mbH_{_{\calS }}^{\H} \mbH_{_{\calS }} + \mbH_{_{\calT}}^{\H} \mbH_{_{\calT }}\succeq \mbH_{_{\calT}}^{\H} \mbH_{_{\calT }}$. This proves the monotonicity of the set function  $f(\calS)$.
We define $\mbH_{_{\calS \cup \{u\}}}=\left[\begin{array}{c|c}
\mbH_{_{\calS }}^{\H} &  \mbH_{_{\{u\}}}^{\H}
\end{array}\right]^{\H}$ and $\mbH_{_{\calT \cup \{u\}}}=\left[\begin{array}{c|c}
\mbH_{_{\calT }}^{\H} &  \mbH_{_{\{u\}}}^{\H}
\end{array}\right]^{\H}$, similar to~\eqref{monotone}  we have
\begin{align}
&f(\calS \cup \{u\}) - f(\calS) -f(\calT \cup \{u\})+f(\calT)\nonumber\\
&=\mathrm{ln}\left(\frac{\mathrm{det}(\mbH_{_{\calS \cup \{u\}}}^{\H} \mbH_{_{\calS \cup \{u\}}}+\mbI)\mathrm{det}(\mbH_{_{\calT}}^{\H} \mbH_{_{\calT}}+ \mbI)}{\mathrm{det}(\mbH_{_{\calS}}^{\H} \mbH_{_{\calS}}+ \mbI)\mathrm{det}(\mbH_{_{\calT \cup \{u\}}}^{\H} \mbH_{_{\calT \cup \{u\}}}+\mbI)}\right)\nonumber\\
&= \mathrm{ln}\left(\frac{\mathrm{det}(\mbH_{_{\calS}}^{\H} \mbH_{_{\calS }}+\mbH_{_{\{u\}}}^{\H} \mbH_{_{\{u\}}}+\mbI)\mathrm{det}(\mbH_{_{\calT}}^{\H} \mbH_{_{\calT}}+ \mbI)}{\mathrm{det}(\mbH_{_{\calS}}^{\H} \mbH_{_{\calS}}+ \mbI)\mathrm{det}(\mbH_{_{\calT }}^{\H} \mbH_{_{\calT}}+\mbH_{_{\{u\}}}^{\H} \mbH_{_{\{u\}}}+\mbI)}\right)\nonumber\\
&=
\mathrm{ln}\left(\frac{\mathrm{det}\left(\mbH_{_{\{u\}}}^{\H}(\mbH_{_{\calS}}^{\H} \mbH_{_{\calS}}+ \mbI)^{-1}\mbH_{_{\{u\}}}\right)}{\mathrm{det}\left(\mbH_{_{\{u\}}}^{\H}(\mbH_{_{\calT}}^{\H} \mbH_{_{\calT}}+ \mbI)^{-1}\mbH_{_{\{u\}}}\right)}\right) \geq 0,
\end{align}
where the second inequality follows from the determinant lemma and the last equality follows from $\calS \subseteq \calT$ that leads to $\mbH_{_{\calT}}^{\H} \mbH_{_{\calT}}\succeq \mbH_{_{\calS}}^{\H} \mbH_{_{\calS}}$ and, hence, $\left(\mbH_{_{\calS}}^{\H} \mbH_{_{\calS}}+\mbI\right)^{-1}\succeq \left(\mbH_{_{\calT}}^{\H} \mbH_{_{\calT}}+\mbI\right)^{-1}$.
\end{proof}
%-------------------------------------------
Algorithm~\ref{alg:place} below summarizes the greedy algorithm for optimally placing $M$ IRS platforms. It has been shown before \cite{shulkind2018sensor,feige1998threshold} that the worst-case  % $1-1/e$ 
optimality guarantee for the performance of the greedy algorithm is evaluated as
 \begin{equation}\label{eq:bound}
 f(\calS^{gr}) \geq  c^{-1}\left(1-\frac{1}{e^{c}}\right)f(\calS^{*})  \geq \left(1-\frac{1}{e}\right)f(\calS^{*}),
 \end{equation}
where $e$ is the Euler's constant, $f(\calS^{*})$ represents the optimal objective value of \eqref{eq:opt3} and  $c\in[0,1]$ is the curvature of the   submodular objective  $f(.)$ \cite{sviridenko2017optimal}:
\begin{equation}
c=1-~\underset{\{j\} \in \calA}{\textrm{minimize}} \quad   \frac{f(\calA)-f(\calA-\{j\})}{f(\{j\})}.
\end{equation}
We compute the tighter bound in~\eqref{eq:bound} numerically and use it as a benchmark to evaluate the performance of the greedy algorithm. The greedy algorithm often outperforms the worst-case lower bounds, as demonstrated in our numerical experiments.

%------------------------------------------------
\begin{algorithm}[H]
\caption{Greedy algorithm for IRS placement.}
    \label{alg:place}
    \begin{algorithmic}[1]
    \Statex \textbf{Input}  $\calA$
    \Statex  \textbf{Output} $\calS^{gr}$
    \State \textbf{Initialization}  $\calS=\emptyset$
    \State  \textbf{While} $|\calS| < M$ \textbf{do} 
     \State \hspace{1cm} $u^{*}\leftarrow\underset{u \in \calA \setminus \calS}{\textrm{argmax}}\quad  f\left(\calS \cup \{u\}\right)$
    \State \hspace{1cm} $\calS \leftarrow \calS \cup
      \{u^{*}\}$
      \Statex  \textbf{Return} $\calS^{gr}\leftarrow  \calS$
   \end{algorithmic}
\end{algorithm}

%----------------------------------
\section{Numerical Experiments}
%This section explores the application of the proposed greedy algorithm for submodular optimization, particularly in the placement of IRS platforms in a radar system. 
We validated our algorithm through numerical experiments in which We set  $N_r=N_t=8$ antennas for the radar transmitter and receiver and the radar was located in the 2-D Cartesian plane at [$0$ m, $0$ m]. The target was at the range $d_{tr}=60$ m and azimuth angle of arrival $\theta_{tr}=\pi/4$. The discretized radar aperture $\calA$ is chosen to be the  Cartesian product of the set $\calR$ with $K_{_{\theta}}=100$ candidate ranges spaced with $\delta=1$ m  in the interval $[0,100]$m and the set $\calM$ comprised of $K_{_{r}}=12$ azimuth angles from the interval $[0,2\pi)$ spaced with $\mu=\pi/6$.

%----------------------------------------------------
\begin{figure}[t]
\centering 
\includegraphics[width=1\columnwidth]{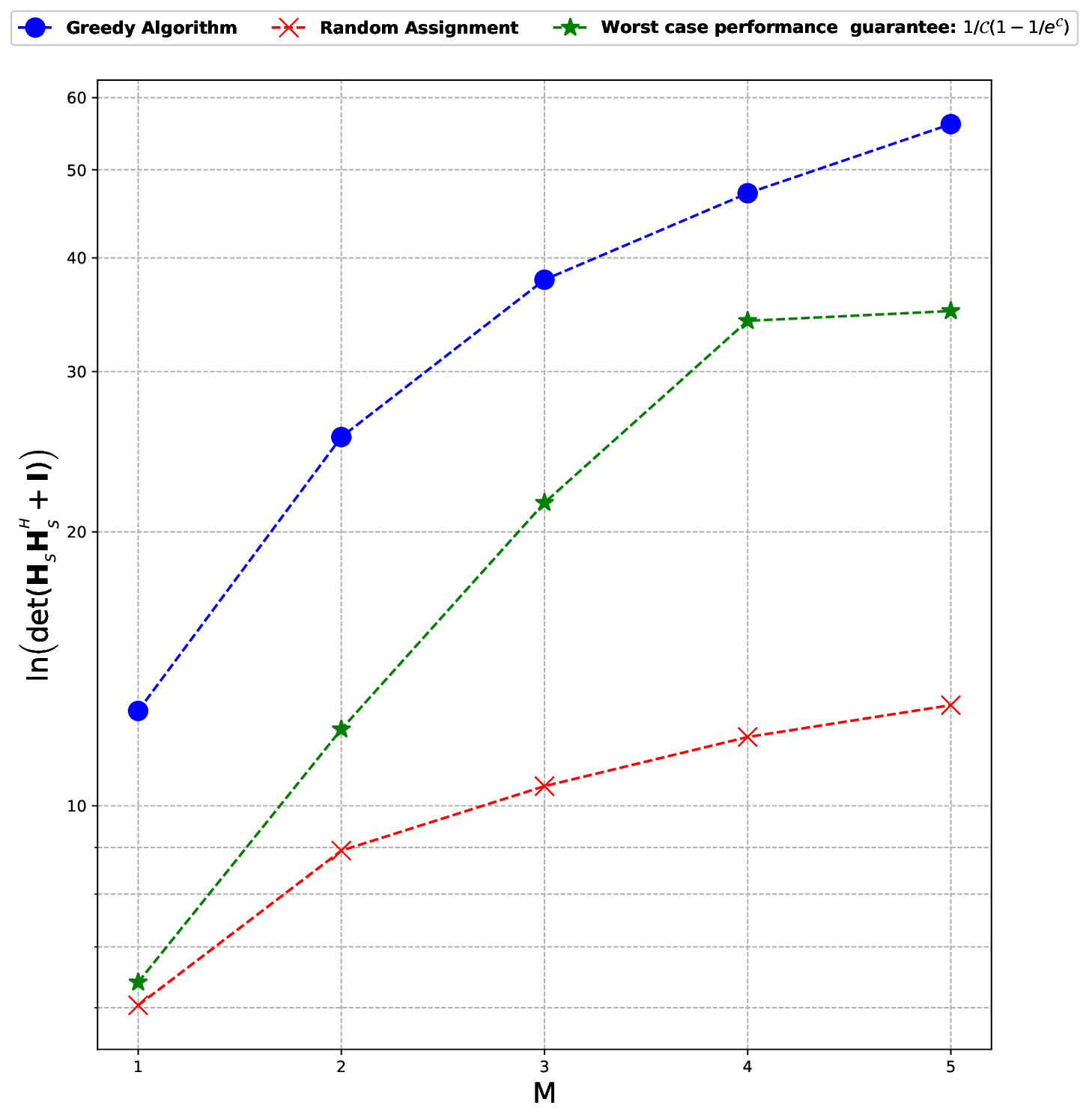}
\caption{The value of cost function of~\eqref{eq:opt3} plotted against the increasing number of IRS platforms $M$ that are sequentially placed in the radar coverage area $\mathcal{A}$.
}
\label{fig::2}
\end{figure}
%----------------------------------------------------
\begin{figure}[t]
\centering 
\includegraphics[width=1.0\columnwidth]{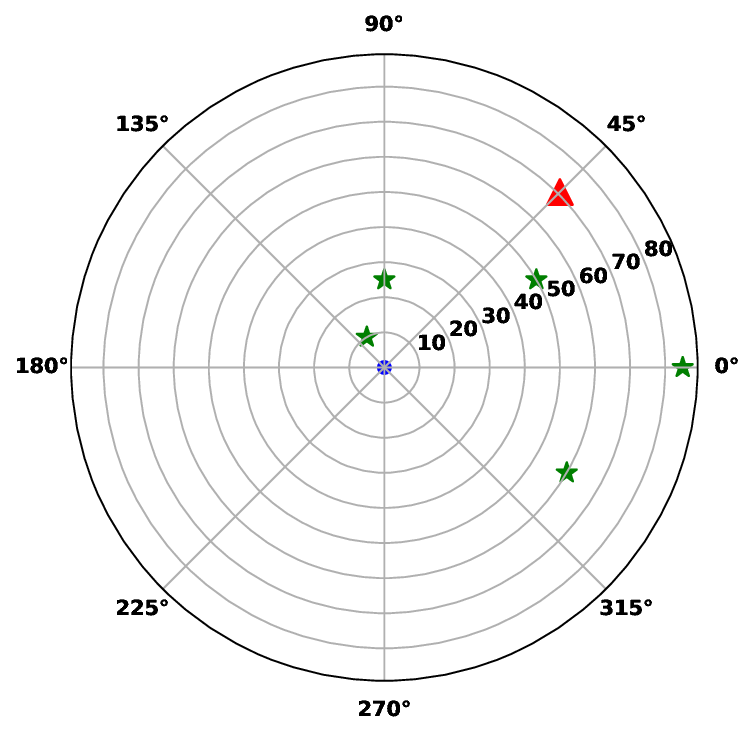}
\caption{Optimal placement of $M=5$ IRS platforms (green star) to illuminate the target (red triangle). The radar (blue circle) is located at the origin.}
\label{fig::3}
\end{figure}
%----------------------------------------------------

We employed Algorithm~\ref{alg:place} to decide the  ranges and azimuth angles $(r_m,\theta_{ri,m}) \in \calA$ of $M$ IRS platforms tuned at randomly generated phase-shifts comprised in  $\bPhi_m$, for  $m \in \{1,\ldots,M\}$. Fig.~\ref{fig::2} compares the optimized MI-based cost function $f(\calS)$ for different numbers $M$ of
deployed IRS platforms with random placement and the worst-case optimality guarantee of the greedy algorithm. In order to 
 compute the worst-case optimality guarantee of the greedy algorithm, we computed the curvature of the submodular objective function defined in~\eqref{eq:bound} through an exhaustive search. Our experiments showed that $c \approx 0.996$. Fig.~\ref{fig::3} illustrates the optimal IRS locations in polar coordinates as obtained by Algorithm~\ref{alg:place}, for $M=5$ IRS platforms.

\section{Summary}
We considered an IRS-assisted sensing system and optimized the placement of IRS platforms using the recently proposed tools of submodular optimization. We showed that our MI-based design criterion is submodular and monotonic. The greedy algorithm for submodular optimization yields the optimal placement of IRS platforms in the radar coverage area with a constant suboptimality guarantee. The IRS-aided radar with optimized IRS locations outperforms random placement and the theoretical worst-case optimality guarantee of the greedy algorithm.

%\clearpage
\bibliographystyle{IEEEtran}
\footnotesize{\bibliography{refs}}

\begin{thebibliography}{10}
\providecommand{\url}[1]{#1}
\csname url@samestyle\endcsname
\providecommand{\newblock}{\relax}
\providecommand{\bibinfo}[2]{#2}
\providecommand{\BIBentrySTDinterwordspacing}{\spaceskip=0pt\relax}
\providecommand{\BIBentryALTinterwordstretchfactor}{4}
\providecommand{\BIBentryALTinterwordspacing}{\spaceskip=\fontdimen2\font plus
\BIBentryALTinterwordstretchfactor\fontdimen3\font minus
  \fontdimen4\font\relax}
\providecommand{\BIBforeignlanguage}[2]{{%
\expandafter\ifx\csname l@#1\endcsname\relax
\typeout{** WARNING: IEEEtran.bst: No hyphenation pattern has been}%
\typeout{** loaded for the language `#1'. Using the pattern for}%
\typeout{** the default language instead.}%
\else
\language=\csname l@#1\endcsname
\fi
#2}}
\providecommand{\BIBdecl}{\relax}
\BIBdecl

\bibitem{elbir2022rise}
A.~M. Elbir, K.~V. Mishra, M.~B. Shankar, and S.~Chatzinotas, ``The rise of
  intelligent reflecting surfaces in integrated sensing and communications
  paradigms,'' \emph{IEEE Network}, 2022, in press.

\bibitem{wei2023ris}
T.~Wei, L.~Wu, K.~V. Mishra, and M.~Shankar, ``{RIS}-aided wideband holographic
  {DFRC},'' \emph{arXiv preprint arXiv:2305.04602}, 2023.

\bibitem{an2023stacked}
J.~An, C.~Xu, D.~W.~K. Ng, G.~C. Alexandropoulos, C.~Huang, C.~Yuen, and
  L.~Hanzo, ``Stacked intelligent metasurfaces for efficient holographic {MIMO}
  communications in 6{G},'' \emph{IEEE Journal on Selected Areas in
  Communications}, vol.~41, no.~8, pp. 2380--2396, 2023.

\bibitem{tyrovolas2022performance}
D.~Tyrovolas, S.~A. Tegos, E.~C. Dimitriadou-Panidou, P.~D. Diamantoulakis,
  C.~K. Liaskos, and G.~K. Karagiannidis, ``Performance analysis of cascaded
  reconfigurable intelligent surface networks,'' \emph{IEEE Wireless
  Communications Letters}, vol.~11, no.~9, pp. 1855--1859, 2022.

\bibitem{wang2023}
F.~Wang and A.~L. Swindlehurst, ``Hybrid {RIS}-assisted interference mitigation
  for spectrum sharing,'' in \emph{IEEE International Conference on Acoustics,
  Speech and Signal Processing}, 2023, pp. 1--5.

\bibitem{esmaeilbeig2023quantized}
Z.~Esmaeilbeig, A.~Eamaz, K.~V. Mishra, and M.~Soltanalian, ``Quantized
  phase-shift design of active {IRS} for integrated sensing and
  communications,'' in \emph{IEEE International Conference on Acoustics,
  Speech, and Signal Processing Workshops}, 2023, pp. 1--5.

\bibitem{wang2023stars}
Z.~Wang, X.~Mu, and Y.~Liu, ``{STARS} enabled integrated sensing and
  communications,'' \emph{IEEE Transactions on Wireless Communications},
  vol.~22, no.~10, pp. 6750--6765, 2023.

\bibitem{wei2022simultaneous}
T.~Wei, L.~Wu, K.~V. Mishra, and S.~M. Bhavani, ``Simultaneous active-passive
  beamformer design in {IRS}-enabled multi-carrier {DFRC} system,'' in
  \emph{European Signal Processing Conference}, 2022, pp. 1007--1011.

\bibitem{wei2023multi}
T.~Wei, L.~Wu, K.~V. Mishra, and M.~Shankar, ``Multi-{IRS}-aided
  {D}oppler-tolerant wideband {DFRC} system,'' \emph{IEEE Transactions on
  Communications}, 2023, in press.

\bibitem{hodge2023}
J.~A. Hodge, K.~V. Mishra, B.~M. Sadler, and A.~I. Zaghloul, ``Index-modulated
  metasurface transceiver design using reconfigurable intelligent surfaces for
  6{G} wireless networks,'' \emph{IEEE Journal of Selected Topics in Signal
  Processing}, 2023, in press.

\bibitem{esmaeilbeig2022irs}
Z.~Esmaeilbeig, K.~V. Mishra, and M.~Soltanalian, ``{IRS}-aided radar:
  {E}nhanced target parameter estimation via intelligent reflecting surfaces,''
  in \emph{IEEE Sensor Array and Multichannel Signal Processing Workshop},
  2022, pp. 286--290.

\bibitem{esmaeilbeig2022joint}
Z.~Esmaeilbeig, A.~Eamaz, K.~V. Mishra, and M.~Soltanalian, ``Joint waveform
  and passive beamformer design in multi-{IRS} aided radar,'' in \emph{IEEE
  International Conference on Acoustics, Speech and Signal Processing}, 2023,
  in press.

\bibitem{esmaeilbeig2022cramer}
Z.~Esmaeilbeig, K.~V. Mishra, A.~Eamaz, and M.~Soltanalian, ``Cram{\'e}r-{R}ao
  lower bound optimization for hidden moving target sensing via
  multi-{IRS}-aided radar,'' \emph{IEEE Signal Processing Letters}, vol.~29,
  pp. 2422--2426, 2022.

\bibitem{esmaeilbeig2023moving}
Z.~Esmaeilbeig, A.~Eamaz, K.~V. Mishra, and M.~Soltanalian, ``Moving target
  detection via multi-{IRS}-aided {OFDM} radar,'' in \emph{IEEE Radar
  Conference}, 2023, pp. 1--6.

\bibitem{Huangplacement2022}
P.-Q. Huang, Y.~Zhou, K.~Wang, and B.-C. Wang, ``Placement optimization for
  multi-{IRS}-aided wireless communications: An adaptive differential evolution
  algorithm,'' \emph{IEEE Wireless Communications Letters}, vol.~11, no.~5, pp.
  942--946, 2022.

\bibitem{zeng2020reconfigurable}
S.~Zeng, H.~Zhang, B.~Di, Z.~Han, and L.~Song, ``Reconfigurable intelligent
  surface ({RIS}) assisted wireless coverage extension: {RIS} orientation and
  location optimization,'' \emph{IEEE Communications Letters}, vol.~25, no.~1,
  pp. 269--273, 2020.

\bibitem{stratidakis2022optimal}
G.~Stratidakis, S.~Droulias, and A.~Alexiou, ``Optimal position and orientation
  study of reconfigurable intelligent surfaces in a mobile user environment,''
  \emph{IEEE Transactions on Antennas and Propagation}, vol.~70, no.~10, pp.
  8863--8871, 2022.

\bibitem{Ghose2023}
S.~Ghose, D.~Mishra, S.~P. Maity, and G.~C. Alexandropoulos, ``{RIS} reflection
  and placement optimisation for underlay {D2D} communications in cognitive
  cellular networks,'' in \emph{IEEE International Conference on Acoustics,
  Speech and Signal Processing}, 2023, pp. 1--5.

\bibitem{ntontin2021optimal}
K.~Ntontin, D.~Selimis, A.-A.~A. Boulogeorgos, A.~Alexandridis, A.~Tsolis,
  V.~Vlachodimitropoulos, and F.~Lazarakis, ``Optimal reconfigurable
  intelligent surface placement in millimeter-wave communications,'' in
  \emph{European Conference on Antennas and Propagation}, 2021, pp. 1--5.

\bibitem{chen2022reconfigurable}
A.~Chen, Y.~Chen, and Z.~Wang, ``Reconfigurable intelligent surface deployment
  for blind zone improvement in mm{W}ave wireless networks,'' \emph{IEEE
  Communications Letters}, vol.~26, no.~6, pp. 1423--1427, 2022.

\bibitem{tohidi2023near}
E.~Tohidi, S.~Haesloop, L.~Thiele, and S.~Stanczak, ``Near-optimal {LoS} and
  orientation aware intelligent reflecting surface placement,'' \emph{arXiv
  preprint arXiv:2305.03451}, 2023.

\bibitem{ulaby1981microwave}
F.~T. Ulaby, R.~K. Moore, and A.~K. Fung, \emph{Microwave Remote Sensing:
  {A}ctive and Passive. {V}olume 1 - {M}icrowave Remote Sensing Fundamentals
  and Radiometry}.\hskip 1em plus 0.5em minus 0.4em\relax Artech House, 1981.

\bibitem{naghsh2013majorization}
M.~M. Naghsh, M.~Modarres-Hashemi, S.~ShahbazPanahi, M.~Soltanalian, and
  P.~Stoica, ``Majorization-minimization technique for multi-static radar code
  design,'' in \emph{European Signal Processing Conference}, 2013, pp. 1--5.

\bibitem{naghsh2014unimodular}
M.~M. Naghsh, M.~Modarres-Hashemi, A.~Sheikhi, M.~Soltanalian, and P.~Stoica,
  ``Unimodular code design for {MIMO} radar using bhattacharyya distance,'' in
  \emph{IEEE International Conference on Acoustics, Speech and Signal
  Processing}, 2014, pp. 5282--5286.

\bibitem{radarsignaldesign2022}
M.~Alaee-Kerahroodi, M.~Soltanalian, P.~Babu, and M.~R.~B. Shankar,
  \emph{Signal design for modern radar systems}.\hskip 1em plus 0.5em minus
  0.4em\relax Artech House, 2022.

\bibitem{liu2020co}
J.~Liu, K.~V. Mishra, and M.~Saquib, ``Co-designing statistical {MIMO} radar
  and in-band full-duplex multi-user {MIMO} communications,'' \emph{arXiv
  preprint arXiv:2006.14774}, 2020.

\bibitem{joshi2009sensor}
S.~Joshi and S.~Boyd, ``Sensor selection via convex optimization,'' \emph{IEEE
  Transactions on Signal Processing}, vol.~57, no.~2, pp. 451--462, 2009.

\bibitem{Godrich}
H.~Godrich, A.~P. Petropulu, and H.~V. Poor, ``Sensor selection in distributed
  multiple-radar architectures for localization: A knapsack problem
  formulation,'' \emph{IEEE Transactions on Signal Processing}, vol.~60, no.~1,
  pp. 247--260, 2012.

\bibitem{mishra2023sparse}
K.~V. Mishra, A.~M. Elbir, and K.~Ichige, ``Sparse array design for direction
  finding using deep learning,'' in \emph{Sparse Arrays for Radar, Sonar, and
  Communications}.\hskip 1em plus 0.5em minus 0.4em\relax Wiley-IEEE Press,
  2023, in press.

\bibitem{shulkind2018sensor}
G.~Shulkind, S.~Jegelka, and G.~W. Wornell, ``Sensor array design through
  submodular optimization,'' \emph{IEEE Transactions on Information Theory},
  vol.~65, no.~1, pp. 664--675, 2018.

\bibitem{tohidi2020submodularity}
E.~Tohidi, R.~Amiri, M.~Coutino, D.~Gesbert, G.~Leus, and A.~Karbasi,
  ``Submodularity in action: {F}rom machine learning to signal processing
  applications,'' \emph{IEEE Signal Processing Magazine}, vol.~37, no.~5, pp.
  120--133, 2020.

\bibitem{mishra2017sub}
K.~V. Mishra and Y.~C. Eldar, ``Sub-{N}yquist channel estimation over {IEEE}
  802.11ad link,'' in \emph{IEEE International Conference on Sampling Theory
  and Applications}, 2017, pp. 355--359.

\bibitem{cover1999elements}
T.~M. Cover, \emph{Elements of information theory}.\hskip 1em plus 0.5em minus
  0.4em\relax John Wiley \& Sons, 1999.

\bibitem{feige1998threshold}
U.~Feige, ``A threshold of ln $n$ for approximating set cover,'' \emph{Journal
  of the ACM}, vol.~45, no.~4, pp. 634--652, 1998.

\bibitem{sviridenko2017optimal}
M.~Sviridenko, J.~Vondr{\'a}k, and J.~Ward, ``Optimal approximation for
  submodular and supermodular optimization with bounded curvature,''
  \emph{Mathematics of Operations Research}, vol.~42, no.~4, pp. 1197--1218,
  2017.

\end{thebibliography}
\end{document}